\documentclass{amsart}

\usepackage{amsmath, amsthm,  amsfonts, amssymb}
\usepackage{latexsym,mathrsfs,yfonts,enumerate}
\RequirePackage{amscd}
\usepackage{array}
\usepackage{hyperref}
\usepackage[all]{xy}
\usepackage{slashbox}
\usepackage{color}
\usepackage{stmaryrd}
\input xy 

\xyoption{all}
\xyoption{2cell} 
\xyoption{v2}

\DeclareMathOperator{\Diff}{Diff}

\theoremstyle{plain}
\newtheorem{theorem}{Theorem}[section]

\newtheorem{proposition}[theorem]{Proposition}

\theoremstyle{definition}
\newtheorem{definition}[theorem]{Definition}

\theoremstyle{remark}

\newtheorem{remark}{Remark}[section]

\numberwithin{equation}{section}
\numberwithin{figure}{section}

\newcommand{\cC}{{\mathcal C}}
\renewcommand{\cH}{{\mathcal H}}
\newcommand{\cG}{\mathcal{G}}
\newcommand{\cS}{{\mathcal S}}

\newcommand{\cE}{{\mathcal E}}

\newcommand{\cU}{{\mathcal U}}
\renewcommand{\cD}{{\mathcal D}}

\newcommand{\CL}{{\mathcal L}}
\newcommand{\cT}{{\mathcal T}}
\newcommand{\cV}{{\mathcal V}}
\newcommand{\cF}{{\mathcal F}}
\newcommand{\cI}{{\mathcal I}}

\newcommand{\CC}{{\mathbb C}}
\newcommand{\RR}{{\mathbb R}}
\newcommand{\ZZ}{{\mathbb Z}}

\newcommand{\TT}{{\mathbb T}}

\newcommand{\ft}{{\mathfrak t}}
\newcommand{\vac}{|0\rangle}
\newcommand{\lb}{\llbracket}
\newcommand{\rb}{\rrbracket}

\newcommand\ev{\text{ev}}

\theoremstyle{plain} 
\newtheorem{thm}{Theorem}[section] 
 
\newtheorem{prop}[thm]{Proposition} 
  
\theoremstyle{definition} 
\newtheorem{defn}{Definition}[section]

\theoremstyle{remark}

\begin{document}

\title[T-duality of current algebras and their quantization]{T-duality of current algebras \\ and their quantization}

\author[P. Hekmati]{Pedram Hekmati}
  \address[P. Hekmati]
  {Department of Pure Mathematics\\
  University of Adelaide\\
  Adelaide, SA 5005 \\
  Australia}
  \email{pedram.hekmati@adelaide.edu.au}

\author[V. Mathai]{Varghese Mathai}
  \address[V. Mathai]
  {Department of Pure Mathematics\\
  University of Adelaide\\
  Adelaide, SA 5005 \\
  Australia}
  \email{mathai.varghese@adelaide.edu.au}
  
\date{\today}

\thanks{The authors acknowledge the support of the Australian Research Council's Discovery 
Project Scheme (under project numbers DP110100072 and DE120102657). We also thank Joel Ekstrand and Reimundo Heluani for bringing their research to our attention.}


\begin{abstract}
In this paper we show that the T-duality transform of Bouwknegt, Evslin and Mathai
applies to determine isomorphisms of certain current algebras and their associated vertex algebras 
on topologically distinct T-dual spacetimes  compactified to circle bundles with $H$-flux.  \\
\\
\center \small \it Dedicated to Steven Rosenberg on the occasion of his $\it{60}^{th}$ birthday \em
\end{abstract}

\maketitle


\section{Introduction}

T-duality is a fundamental symmetry in string theory, which in particular gives an equivalence 
between type IIA and IIB string theories on spacetimes that are compactified in one spatial direction. The 
duality in this case is simply an interchange of the radius $R\leftrightarrow 1/R$, or more precisely the Fourier transform in the circle direction. The relationship between T-dual manifolds in a topologically trivial 
$H$-flux was first worked out using non-linear sigma models by Buscher in \cite{Bu} 
and was further  elaborated upon by
Ro\v{c}ek and Verlinde in \cite{RV}.  In the presence of a  topologically
non-trivial integral 3-form flux on spacetimes that are compactified 
as circle bundles, the topology and background flux of the
T-dual spaces were determined for the first time by
Bouwknegt, Evslin and Mathai in~\cite{BEM1,BEM2}. There the authors also established an isomorphism 
between the charges of the Ramond-Ramond fields on spacetime and its T-dual partner in twisted K-theory and twisted
cohomology.

In ~\cite{CG}, Cavalcanti and Gualtieri showed that the T-duality 
transformations in \cite{BEM1,BEM2} can be understood in
the framework of generalized geometry introduced by Hitchin  \cite{Hi03} and 
developed by Gualtieri \cite{Gu03}.  
The title role in this geometry is played by the generalized tangent bundle. This is the direct sum of the tangent and cotangent
bundles of a manifold and comes equipped with a canonical
orthogonal structure and a Courant bracket. The latter depends, up to isomorphisms, on the
choice of a real degree three  cohomology class of the manifold. The analogy with the ordinary Lie bracket also holds up to skew-symmetry only, since the Courant bracket violates the Jacobi identity. Generalized geometry has the very nice feature that it subsumes complex and symplectic geometry as particular extremal limits.  In \cite{CG} it was shown that T-duality can be viewed as
an isomorphism between the underlying orthogonal  and Courant 
algebroid structures of the T-dual circle bundles, where the $H$-flux now plays the role of twisting the Courant brackets. 

In the present paper we use this perspective on T-duality, as well as the observation 
of Alekseev and Strobl \cite{AS},  which maps a particular current algebra 
to the Courant algebroid. 
These current algebras appear as Noether symmetries of certain sigma models that 
are of fundamental importance in string theory.
We show that T-duality 
gives an isomorphism of an invariant version 
of the Alekseev-Strobl current algebras on T-dual spacetimes that are circle bundles with $H$-flux. 
We  enhance this further to an isomorphism of the associated universal vertex algebras \cite{K} on the 
T-dual pair. A key observation here is the two ways in which the $H$-flux can be transported to the loop space.
This is achieved either by transgression, which determines a twisted symplectic 2-form on the 
cotangent bundle of the loop space, or by looping the flux to an integral 3-form on the loop space. The relationship
between the algebraic structures associated to these forms was discovered in \cite{AS} and is clearly stated and derived in Proposition \ref{prop2}. Similar kinds of T-duality isomorphism of vertex algebras and generalizations thereof have recently been established in some special backgrounds, see for instance \cite{AH, KO}.

For higher rank torus bundles, an analogous T-duality procedure
 works provided that the integral 3-form satisfies $i_V i_W  H=0$, where $V,W$ are vector fields tangent to the fibres, as was shown in \cite{BHM, MR1, MR2, MR3}.
That is, one can consistently iterate the T-duality procedure in this case one circle at a time. 
The higher rank case was also carried out in the context of generalized geometry and 
(twisted) Courant algebroids in this context in \cite{CG}. 
Relaxing the restriction above on the $H$-flux, it was shown by
Mathai and
Rosenberg~\cite{MR1, MR2}  that 
the T-dual manifold may be viewed as a noncommutative space, which is a $C^*$-bundle with fibres 
that are (stabilized) noncommutative tori.  It is an open problem 
to explore the analogue of this in terms of generalizations of current algebras and their quantizations. 

The first six sections of the paper reviews the literature in a form that is suitable for our context.
Our first main result is Proposition \ref{thm1} which is the bottom horizontal arrow in the diagram below, 
$$
\xymatrix@C7pc{
\text{Poisson Algebra} \ar[r]_{\text{Inclusion}\qquad\qquad} \ar[d]_{\text{Fourier Transform}}
& \text{Lie Algebra of Local Distributions}  \ar[d]^{\text{Taylor Expansion}} \\
\text{Poisson Vertex Algebra}  \ar[r]_{\text{Proposition 7.1}\quad}  &  \text{Weak Courant Dorfman Algebra} 
}
$$
It says that any Poisson vertex algebra gives rise to a weak Courant Dorfman algebra. The latter was introduced by
Ekstrand and Zabzine \cite{E, EZ} by relaxing certain axioms of a Courant-Dorfman algebra \cite{R}. It was further shown in \cite{EZ} that such a structure is naturally induced on the space of local functions by the Lie bracket on currents (or local distributions).  Proposition \ref{prop1} asserts that the two different ways of constructing a weak Courant Dorfman algebra
from a Poisson algebra are equivalent, or in other words, the diagram above commutes.
The last two sections contain our other main results, which detail the T-duality isomorphisms of 
Alekseev-Strobl current algebras and of their associated universal vertex algebras.

This paper is an initial step in our program towards establishing a T-duality isomorphism
for more general current algebras
and their quantizations on  toroidally compactified spacetimes with $H$-flux.

\section{Cotangent bundle of the loop space}
Let $E$ be a smooth manifold and denote by $LE = C^\infty(S^1, E)$ the space of smooth parametrized loops
endowed with the standard structure of a Fr\'echet manifold. Recall that a tangent vector to a 
loop $x \in LE$ is a vector field along the map $x(t)$, namely
$T_x LE = \Gamma(S^1, x^* TE)$ is the space of smooth sections of the pullback bundle $x^*TE$. 
In fact, there is a natural diffeomorphism of the manifolds $TLE$ and $LTE$ which covers the identity on $LE$ \cite{S}. 
The following map
\begin{equation}\label{inclusion}
LE \to TLE, \ \ x(t) \mapsto \partial x(t)=(x_*)\left(\frac{d}{dt}\right)(t)
\end{equation}
defines an inclusion of the loop space into its tangent bundle, where $t\mapsto \left(t, \frac{d}{dt}\right)$ is
the natural section of $TS^1\to S^1$.
The group of orientation preserving diffeomorphisms of the circle $\Diff_+(S^1)$ acts smoothly on
the loop space by precomposition,
$$ \Diff_+(S^1) \times LE \to LE, \ \ (\phi, x) \mapsto x \circ \phi $$
and its fixed point set is the space $E$ of constant loops.

The cotangent bundle of the loop space $T^*LE$ is the phase space of non-linear sigma models on the cylinder $\Sigma = S^1\times \RR$ with target space $E$. However unlike $TLE$, the definition of the cotangent bundle is more subtle. While $LT^*E$ is modelled on 
the vector space $L\RR^n$, the process of dualizing means that the model space for $T^*LE$ is the dual space 
$(L\RR^n)^*$ consisting of $\RR^n$-valued distributions on the circle. However, the inclusion 
$$LT^*E = \text{Hom}_{L\RR}(LTE, L\RR)\subset \text{Hom}_\RR(LTE, L\RR) \xrightarrow{\cE} \text{Hom}_\RR(LTE,\RR)=T^*LE $$ 
induced by the $S^1$-equivariant functional $\cE \colon L\RR \to \RR, \ f \mapsto \cE(f)=\int_{S^1}fdt$  is injective. 
There is further a natural pairing
$$T_x LE  \times L_x T^*E  \to C^\infty(E), \ \ (\xi,\alpha ) \to \int_{S^1} \langle \xi(t), \alpha(t) \rangle dt $$
induced by the pairing $\langle \cdot, \cdot\rangle\colon TE \times T^*E \to C^\infty(E)$. We shall therefore adopt $LT^*E$ as our definition of the cotangent bundle $T^*LE$, rather than $(TLE)^*$. 

Let $q:T^*LE\to LE$ denote the projection map. The cotangent bundle carries a canonical 1-form $\eta\in \Omega^1(T^*LE)$ defined by
$$ \eta_{\alpha}(\chi) =  \int_{S^1} \langle (q_*)_{\alpha}(\chi)(t),  \alpha(t)\rangle dt \ ,$$
where $\alpha \in T^*_x LE$ and $\chi\in T_{\alpha}(T^*LE)$. Its differential determines the canonical symplectic form $\omega = d\eta \in \Omega^2_\ZZ(T^*LE)$,  which in local Darboux coordinates $(x^i,p_i)$ on $T^*E$ takes the familiar form
$$ \omega = \sum_{i=1}^N\int_{S^1} dp_i(t)\wedge dx^i(t) dt\ ,$$
where $N= \text{dim}\ E$. Here $d:\Omega^\bullet(T^*LE) \to  \Omega^{\bullet+1}(T^*LE) $ denotes the usual de Rham differential.
Since $\omega$ is exact, the associated line bundle $\CL_\omega$ over $T^*LE$ is a trivial bundle with connection
$$ \nabla = d + \eta \ .$$

We conclude with some remarks on the topology of the loop space. Namely if $E$ is an H-space,  then $LE \simeq \Omega E \times E$ where $\Omega E$ denotes the space of loops based 
at the identity element in $E$, and it follows that
$$\pi_*(LE) = \pi_{*+1}(E)\times \pi_*(E) \ .$$
More generally, when $E$ is a smooth connected manifold with base point $x_0$, the fibration sequence 
$$\Omega E \to LE \xrightarrow{\ev_0} E$$
is locally trivial and it splits by the embedding of $E$ into $LE$ as the subspace of constant loops. Here $\ev_0$ denotes the evaluation of loops at $t=0$.

\section{Transgression}

Transgression determines a homomorphism in integral cohomology
$$\tau \colon H^n(E,\ZZ) \to H^{n-1}(LE,\ZZ) \ ,$$
which on the level of differential forms is given by pullback along the evaluation map,
$$\ev\colon S^1 \times LE \to E, \ \ (t, x) \mapsto x(t) \ ,$$
followed by fibre integration over the circle,
$$\tau\colon \Omega^\bullet(E) \to \Omega^{\bullet-1}(LE), \ \ \alpha \mapsto \tau(\alpha) = \int_{S^1}\ev^* (\alpha) \ .$$
This map is invariant under the action of $\Diff_+(S^1) \times \Diff_{\alpha}(E)$, where $\Diff_{\alpha}(E)$
denotes the $\alpha$-preserving diffeomorphisms of $E$, and it commutes
with exterior differentiation,
$$d_{LE}\circ \tau = \tau \circ d_E.$$

Let $H\in \Omega^3_\ZZ(E)$ denote a closed differential 3-form on $E$ with integral periods and let $\cG_H$ be an associated bundle gerbe with $H$ as its Dixmier-Douady invariant \cite{M}. A geometric realisation of the cohomological transgression homomorphism
$$\tau \colon H^3(E,\ZZ) \to H^2(LE,\ZZ)  $$
 has been described by Brylinski and McLaughlin in the language of sheaves of groupoids \cite{B, BM}, and more recently by Waldorf in the framework of bundle gerbes \cite{W}. In fact, the transgression can be refined to a homomorphism in Deligne cohomology, sending a bundle gerbe $\cG_H$ with connective structure on $E$ to a principal $\CC^*$-bundle $\mathcal{L}_H$ with connection on $LE$. 

Following the description  in \cite{W}, the fibre of $\CL_H$ over a loop $x\in LE$ consists of the set of isomorphism classes of flat trivialisations of the pullback bundle gerbe $x^* \cG_H$ on $S^1$. Identifying the elements of $\CC^*$ with principal $\CC^*$-bundles with flat connection over $S^1$, via $z \mapsto P_z$ such that Hol$_{P_z}(S^1) =z$, the right $\CC^*$-action on the fibres is given by
$$\CL_H \times \CC^* \to \CL_H, \ \ ([\cT],z) \mapsto [\cT\otimes P_z] \ .$$
The transgressed connection $\theta$ on $\CL_H$ can be characterised by its parallel transport. Namely, consider a path $p:[0,1]\to LE$ and denote by $\cT_0 \in  \CL_{H,x_0}$ and $\cT_1 \in  \CL_{H,x_1}$ the trivialisations of $\cG$ at the end-loops of the path. The parallel transport $P_{p,\theta}:\CL_{H,x_0} \to \CL_{H,x_1}$ is determined by the bundle gerbe holonomy
$$\text{Hol}_{\cG_H}(p, \cT_0, \cT_1)$$
over the associated cylinder $p: S^1\times \RR \to E$, where the  fibre elements  $\cT_0, \cT_1$ at the endpoints determine the boundary conditions. Locally on  a coordinate chart $U_i\subset E$, the pullback  of the connection 1-form $\theta_i$  to $LU_i$ is given by the transgressed curving $B_i\in \Omega^2(U_i)$ on $\cG_H$,
$$ \theta_i = \int_{S^1} \ev^* B_i \ ,$$
where $dB_i = H|_{U_i}$. The curvature of $\CL_H$ on the other hand is described globally by the transgressed flux,
$$\tau(H)_x(\xi,\chi) =\int_{S^1} \iota_{\partial x(t)}H_{x(t)}(\xi(t),\chi(t))dt$$
where $\xi,\chi \in T_x LE$. Moreover, if $\CL_{\omega}$  denotes the line bundle associated to the canonical symplectic structure on $T^*LE$, then the product bundle
$$\CL_{\omega_H} = \CL_{\omega} \otimes q^*\mathcal L_H \ ,$$
has  as its curvature the twisted symplectic form $\omega_H\in \Omega^2_\ZZ(T^*LE)$ given by 
\begin{equation}\label{symp}
\omega_H = \omega +q^* \tau(H) \ ,
\end{equation}
where as earlier $q: T^*LE\to LE$ denotes the projection map. The $L^2$-sections of this pre-quantum line bundle correspond to the wave functions of the quantised sigma model. 

\section{Space of local functionals}

Let us regard the cotangent bundle $T^*LE$ as sitting inside the space of smooth global sections of the trivial bundle $X:=T^*E\times S^1\to S^1$ and denote by $J^\infty (X)\to S^1$ the associated infinite jet bundle. Recall that points on $J^\infty (X)$ are equivalence classes of smooth sections of $X$ whose Taylor coefficients coincide to all orders. Since the fibered manifold $X$ is a product, it follows that the infinite jet bundle is also trivializable,
$$J^\infty (X) \cong S^1 \times \cH\ ,$$
where the fibre $\cH$ is an infinite dimensional vector space. In  local canonical  coordinates $(x^i, p_i)$ on $T^*E$, the induced coordinates on the infinite jet bundle are given by $ (t, x^i,p_i,\partial x^i, \partial p_i, \dots)$.

Since $J^\infty(X)$ is obtained as an inverse limit of topological spaces $\{ J^k(X)\}$, where $J^k(X)$ is the fibre bundle of $k$-jets of  smooth sections of $X$, there exist natural projections $\pi^k\colon J^\infty(X) \to J^k(X)$ for all $k\in \ZZ_+$. 
\begin{defn} A smooth function $f \in C^\infty(J^\infty (X))$ is called \emph{local} if it factorizes as $ f = \tilde f \circ \pi^k$ for some $\tilde f\in C^\infty(J^k (X))$ and some non-negative integer $k$. Let $\cV_{loc} $ denote the subspace of  local functions on the infinite jet bundle.

\end{defn}
In other words, local functions on $J^\infty (X)$ are pullbacks of smooth functions on a finite jet bundle. On open charts, the local functions thus only depend on finitely many loop derivatives of the coordinate functions. We note that $\cV_{loc}$ is a unital, commutative, associative algebra under pointwise multiplication.

In order to define local functionals, let us introduce a differential on the space of local functions,
$$d_h
\colon \cV_{loc} \to \cV_{loc} \otimes \Omega^1(S^1), \ \ f \mapsto \partial f dt \ ,$$
where $\partial f$ is the total derivative in the loop direction. In any system of local coordinates $\{u_i\}_{i\in \cI}$ on $T^*E$ where $\cI=\{1,2,\dots, 2N\}$, it takes the familiar form
$$d_h f = \left(\partial_t f + \sum_{i\in \cI, m\in \ZZ_+} u^{(m+1)}_i \frac{\partial f} {\partial u^{(m)}_i} \right) dt \ ,$$
with $ u^{(m)}_i =\partial_t^m u_i$. The meaning of this differential will be discussed further in the next section. 
\begin{defn}The space of \emph{local functionals} is defined as the  quotient
$$ \cF^0 := \left(\cV_{loc} \otimes \Omega^1(S^1) \right)/ d_h \cV_{loc} \ .$$
\end{defn} 
Elements $ [f]\in \cF^0$ can be written unambiguously as an integral $[f] = \int_{S^1} f dt$, since we are not dealing with boundary conditions. For later purposes, we also introduce a notion of distributions.
\begin{defn}The space of \emph{local distributions} $\cD$ is defined as the space of bilinear maps,
$$ J\colon\cV_{loc} \times L\RR \to \cF^0, \ \ (f,\phi) \mapsto J_\phi(f) = [f\phi]$$
where $\phi \in L\RR  = C^\infty(S^1,\RR)$ denotes a test function. 
\end{defn} 
Notice that there is a natural inclusion $\cF^0 \subset \cD$ by restricting to the constant test function $\phi = 1$. We remark further that unlike the space of local functions, $\cF^0$ and $\cD$ are not associative algebras.

\section{Variational bicomplex}

The differential $d_h$ introduced in the previous section occupies a natural place in the variational bicomplex which we now briefly review. The space of differential forms $\Omega^\bullet(J^k(X))$ on the $k$-jet bundle  are defined as sections of the exterior algebra bundle $\Lambda^\bullet(J^k(X))$. These form a direct limit system whose direct limit defines the vector space $\Omega^\bullet(J^\infty(X))$. Here and henceforth we shall implicitly restrict to smooth  local functions, so that $\Omega^0(J^\infty(X)) = \cV_{loc}$. The de Rham complex $(\Omega^\bullet(J^\infty(X)), d)$  on the infinite jet bundle contains a differential contact ideal $\cC(J^\infty(X))$ generated by 1-forms $\theta$ which satisfy $j^\infty(\sigma)^*\theta = 0$, where $j^\infty(\sigma)$ is the point on $J^\infty(X)$ associated to the smooth section $\sigma \in\Gamma(X)$. Locally these contact 1-forms can be written
$$\theta_i^m = du^{(m)}_i - u_i^{(m)}dt $$
and they give a meaning to vertical differential forms on $J^\infty(X)$. Likewise there is a notion of horizontal differential forms with components only along the base manifold, consult \cite{A} for a precise definition. Together they induce a splitting of the de Rham complex into a bicomplex
$$ \Omega^p(J^\infty(X)) = \bigoplus_{q+r = p} \Omega^{q,r}(J^\infty(X))$$
with the differential $ d=d_h + d_v $. The differential in the previous section thus corresponds to
$$d_h\colon \Omega^{0,0}(J^\infty(X)) \to \Omega^{1,0}(J^\infty(X)), \ \ f\mapsto \partial f dt \ .$$ 
More generally, we have the \emph{augmented variational bicomplex}
\\

\mbox{ $\begin{CD}
@.  @. @Ad_vAA @Ad_vAA  @A\delta AA   \\
@. 0 @>>> \Omega^{0,2} @>d_h>> \Omega^{1,2} @>I >>\mathcal{F}^2 @>>> 0\\
@.  @. @Ad_vAA @Ad_vAA  @A\delta AA   \\
@. 0 @>>> \Omega^{0,2} @>d_h>> \Omega^{1,2} @>I >>\mathcal{F}^2 @>>> 0\\
@.  @. @Ad_vAA @Ad_vAA  @A\delta AA   \\
@. 0 @>>> \Omega^{0,1} @>d_h>> \Omega^{1,1}  @>I>>\mathcal{F}^1 @>>> 0 \\
@. @.       @Ad_vAA    @Ad_vAA @A\delta AA   @. \\
0 @>>> \RR @>>> \Omega^{0,0} @>d_h>> \Omega^{1,0} @>\int >> \cF^0  @>>> 0 @.\\
 \end{CD}$} \hspace*{-10mm}\label{Bik}
 \\
 \\
 
\noindent where the surjections $I$ are the so called \emph{interior Euler operators} \cite{A} and $\int\colon  \cV_{loc} \otimes \Omega^1(S^1) \to \cF^0$ is the canonical projection map. In other words, the space of local functionals can be identified with the cohomology group
$$ \cF^0 \cong H^1(\Omega^{(*,0)}, d_h) \ .$$
In fact here all the rows are exact, so the space of local functional forms $\cF^s$ is given by 
$$ \cF^s(J^\infty(X))  = \Omega^{1,s}(J^\infty(X))/d_h(\Omega^{0,s}(J^\infty(X))) \ ,$$
for all $s\in \ZZ_+$. The vertical edge complex
$$ \cF^0 \xrightarrow{\delta} \cF^1\xrightarrow{\delta} \cF^2\xrightarrow{\delta} \dots $$
is called the \emph{variational complex} and  it can be realised algebraically as the complex of an abelian Lie conformal algebra with coefficients in the non-trivial module $\cV_{loc}$, \cite{DSHK}. Furthermore, the first differential
$$ \delta\colon \cF^0 \to \cF^1, \ \ [f]=\int_{S^1} f(t,u_i, u_i^{(1)},\dots,u^{(k)}_i)dt \ \mapsto  \ \delta[f]= \frac{\delta}{\delta u} f $$
is given by the familiar Euler-Lagrange derivative in variational calculus,
$$\frac{\delta}{\delta u} f = \sum_{i\in \cI}\frac{\delta f}{\delta u_i}  \theta_i\wedge dt= \sum_{\substack{i\in \cI \\ m\in \ZZ_+}}\left((-\partial)^m \frac{\partial f} {\partial u_i^{(m)}} \right)\theta_i\wedge dt\ .$$
This expression is well-defined on local functionals due to $\frac{\delta}{\delta u}\circ \partial = 0$. Also notice that it is somewhat misleading to call $\frac{\delta}{\delta u}$ a derivative as it violates the Leibniz rule. 
Applying the functor $\text{Hom}(\bullet \times L\RR,\cF^0)$ to the bottom horizontal complex, one has
$$Hom(\cF^0\times L\RR,\cF^0) \xrightarrow{d} Hom(\Omega^{1,0} \times L\RR,\cF^0) \xrightarrow{d} \cD\xrightarrow{d} Hom(\RR \times L\RR,\cF^0) $$
where  $\cD = Hom(\Omega^{0,0}  \times L\RR,\cF^0)$. This is to be interpreted as a topological $Hom$-functor, but we omit the details since it will not be used in the sequal.

\section{Poisson algebra of local functions}

In the remainder of the paper, we shall restrict our attention to local functions that do not explicitly depend on the loop coordinate. These correspond to local functions on the infinite jet space $J^\infty(LT^*E) \subset J^\infty(X)$,  consisting of jets of smooth maps from the circle into $T^*E$ rather than sections of $X$. By abuse of notation, we shall still denote this subspace by $\cV_{loc}$ and the associated spaces of local functionals and local distributions by $\cF^0$ and $\cD$ respectively.

The twisted symplectic structure on the  phase space $\omega_H \in \Omega^2_\ZZ(T^*LE)$ determines in the usual way a Poisson bracket on smooth functions $f,g\in C^\infty(T^*LE) = C^\infty(J^0(LT^*E))$ by
$$\{f,g\} = \omega_H(X_f, X_g) $$
where the Hamiltonian vector field $X_f$ is defined by $\iota_{X_f}\omega_H = df$. Next one would like to  extend this bracket to the space of local functions inside $ C^\infty(J^\infty(LT^*E)) $ and subsequently to all local functionals $\cF^0$ and local distributions $\cD$.  

For an invariant description of the Poisson bracket on local functionals see for instance \cite{BFLS}. Here we simply note that any skew-symmetric bilinear map $B \colon \cF^1\times \cF^1 \to \Omega^{1,0}$ determines a bracket between local functionals $[f],[g] \in \cF^0$ by 
$$ \{[f],[g]\} = [B(\delta [f], \delta [g])] \ .$$
For a Lie bracket, the Jacobi identity imposes the additional condition
$$ [B(\delta [B(a, b)],c)] + [B(\delta [B(b, c)],a)] + [B(\delta [B(c, a)],b)]=0$$ 
for all $a,b,c\in \cF^1$. Similarly, one obtains a Lie bracket on the space of local distributions $\cD$ by setting
$$ \{J_\phi(f),J_\psi(g)\} = [B(\delta [f\phi], \delta [g\psi])] \ .$$

In a system of local coordinates $\{u_i\}_{j\in\cI}$ on $T^*E$, the Poisson bracket on the coordinate functions can be extended to all real analytic local functions $f,g\in \cV_{loc}$ by repeatedly applying the Leibniz rule and bilinearity, 
\begin{equation}\label{poisson}
 \{f(t),g(s)\} = \sum_{\substack{i,j\in \cI \\ m,n \in \ZZ_+}} \frac{\partial f}{\partial u_i^{(m)}} \frac{\partial g}{\partial u_j^{(n)}} \partial^m_t\partial^n_s \{u_i(t),u_j(s)\} \ .
 \end{equation}
In order to extend the bracket to the space of local functionals $\cF^0$, we shall consider Poisson brackets that are `local' in the loop direction,
$$ \{u_i(t),u_j(s)\} = \omega_H(X_{u_i}, X_{u_j}) \delta(t-s) \ ,$$
where the coordinate functions are naturally identified with distributions. This gives rise to the following Lie bracket on $\cF^0$,
$$ \{[f],[g]\} =\left[ \sum_{i,j\in \cI}\frac{\delta g}{\delta u_j}\omega_H(X_{u_i}, X_{u_j}) \frac{\delta f}{\delta u_i}  \right]$$
after performing integration by parts and eliminating the $\delta$-function. This means that locally the associated map $B \colon \cF^1\times \cF^1 \to \Omega^{1,0}$ is given by
$$ B = \sum_{i,j\in \cI} \omega_H(X_{u_i}, X_{u_j}) \tilde\theta^i\wedge \partial_t \wedge \tilde\theta^j = \sum_{i,j\in \cI} B_{ij} \tilde\theta^i\wedge \partial_t \wedge \tilde\theta^j$$
where $\tilde\theta^i(\theta_j)=\delta^i_j$. Moreover, the Lie bracket on the space of local distributions $\cD$ takes the form 
$$ \{J_\phi(f),J_\psi(g)\} = \left[ \sum_{i,j\in \cI}\frac{\delta (g\psi)}{\delta u_j} B_{ij} \frac{\delta (f\phi)}{\delta u_i}  \right]$$
In Darboux coordinates $(x^i,p_i)$ on $T^*E$, a straightforward calculation shows that the local Poisson bracket induced by the twisted symplectic structure \eqref{symp} is given by
\begin{eqnarray} \label{pb}
\{x^i(t), x^j(s)\}&=& 0, \nonumber\\ 
\{x^i(t), p_j(s)\}&=&\delta^i_j\delta(t-s),  \\
\{p_i(t), p_j(s)\}&=&-\sum_{k=1}^N H_{ijk}(t)\partial x^k(t)\delta(t-s). \nonumber
\end{eqnarray}
 
\section{Algebraic structures on the space of local functions}
In this section we explain how the Poisson bracket on $\cV_{loc}$ determines the structure of a Poisson vertex algebra and a weak Courant-Dorfman algebra on the space of local functions.

\subsection{Poisson vertex algebra} Following \cite{BDSK}, 
we note that the Fourier transform of the Poisson bracket of local functions
\begin{equation}\label{fourier}
 \{f(s)_\lambda g(s)\} = \int_{S^1} e^{\lambda (t-s)}\{f(t),g(s)\}dt \ ,
\end{equation} 
can locally be written  in the form
\begin{equation} \label{l-bracket}
\{f(s)_\lambda g(s)\} =\sum_{\substack{i,j\cI \\ m,n\in\ZZ_+}} \frac{\partial g(s)}{\partial u_j^{(n)}}(\partial + \lambda)^n \{u_{i\partial+\lambda}u_j\}_\rightarrow
(-\partial - \lambda)^m\frac{\partial f(s)}{\partial u_i^{(m)}} \ ,
\end{equation} 
where $\{u_{i\partial+\lambda}u_j\}_\rightarrow$ means that the powers of $\partial+\lambda$ are moved to the right when the bracket is expanded in the Fourier parameter. The bracket \eqref{l-bracket} determines a \emph{Poisson vertex algebra} structure on the space of local functions\footnote{More precisely, it determines a sheaf of Poisson vertex algebras on $T^*E$.}. Namely, $\cV_{loc}$
is an associative, commutative, unital algebra endowed with a derivation $\partial: \cV_{loc}\to \cV_{loc}$ 
and an $\RR$-linear $\lambda$-bracket $\{\cdot_\lambda \cdot\}\colon \cV_{loc} \times \cV_{loc} \to \cV_{loc}\otimes \RR[\lambda]$ satisfying
the following axioms:
\\
 
\begin{itemize} 
\item[]\emph{(Sesquilinearity)} \ \ \  $\{\partial f_\lambda g\} = -\lambda \{ f_\lambda g\} , \ \ \{f_\lambda \partial  g\} = (\partial +\lambda) \{ f_\lambda g\}  $
\item[]\emph{(Skew-symmetry)}\ \  \ $\{ f_\lambda g\} = - \{ g_{-\partial-\lambda}f\} $
\item[]\emph{(Jacobi identity)}  \ \ \ $\{ f_\lambda \{g_\mu h\} \}= \{ g_\lambda \{f_\mu h\}  +\{ f_\lambda g\}_{\lambda +\mu} h\} $
\item[]\emph{(Leibniz identity)} \ \   $\{ f_\lambda gh\} =  \{ f_\lambda g\}h + g \{ f_\lambda h\}$\\
\end{itemize} 
By combining skew-symmetry and the Leibniz rule, one obtains the \emph{right} Leibniz identity
$$\{ fg_\lambda h\} = \{ f_{\partial+\lambda} h\}_\rightarrow g + \{ g_{\partial+\lambda} h\}_\rightarrow f \ .$$
Here we have identified $\Omega^{0,0}\simeq \Omega^{1,0}$ under the natural map
$$ \cV_{loc} \to \cV_{loc}\otimes \Omega^1(S^1), \ \ f\mapsto f dt \ .$$
The differential $d_h$  corresponds to a derivation $\partial$ on $\cV_{loc}$ under this identification. 
In terms of the $\lambda$-bracket, the Lie brackets on the space of local functionals and local distributions take a particularly nice form, 
$$ \{[f],[g]\} = [\{f_\lambda g\}]|_{\lambda =0} \ , \ \ \   \{J_\phi(f),J_\psi(g)\} =  [\{f\phi_\lambda g\psi\}]|_{\lambda =0}  \ .$$
We also remark that any Poisson vertex algebra has an underlying \emph{Lie conformal algebra} structure obtained by suppressing  the multiplicative structure on $\cV_{loc}$ and   thus relaxing the Leibniz identity.

\subsection{Weak Courant-Dorfman algebra}
Let us recall the axioms of a \emph{weak Courant-Dorfman algebra} as introduced in the appendix of \cite{EZ}. Namely it is defined by a quintuple $(V, R, \langle \cdot,\cdot \rangle, \lb\cdot,\cdot\rb, \partial)$, where $V$ and $R$ are vector spaces, $\langle \cdot,\cdot \rangle\colon V\otimes V \to R$ is a symmetric bilinear form, $\lb\cdot,\cdot\rb\colon V\otimes V\to V$ is the Dorfman bracket and $\partial\colon R\to V$ is map, subject to relations

\begin{enumerate} 
\item$\lb f,\lb g,h\rb\rb=\lb \lb f,g\rb , h\rb +\lb g, \lb f,h\rb \rb $
\item $\lb f,g\rb + \lb g,f\rb =\partial  \langle f,g\rangle$
\item $\lb \partial a, f\rb = 0$
\end{enumerate}
for all $f,g,h \in V$ and $a\in R$. By promoting $R$ to a commutative algebra, $V$ to a left $R$-module, $\partial$ to a derivation and imposing the following additional axioms,

\begin{enumerate} 
\item[(4)]  $\lb f, ag\rb= a\lb f,g\rb +\langle f,\partial a\rangle g$
\item[(5)] $\langle f,\partial \langle g,h\rangle\rangle=   \langle \lb f,g\rb ,h\rangle + \langle g,\lb f,h\rb \rangle$
\item[(6)] $\langle \partial a, \partial b\rangle =0$
\end{enumerate} 
one recovers the original definition of a Courant-Dorfman algebra by Roytenberg \cite{R}. Actually, a weaker version of the latter two axioms already follows from $(1)-(3)$,
\begin{enumerate} 
\item[(7)] $\partial\big(\langle f,\partial \langle g,h\rangle\rangle -   \langle \lb f,g\rb ,h\rangle - \langle g,\lb f,h\rb \rangle\big) =0$
\item[(8)] $\partial\big(\langle \partial a, \partial b\rangle\big) =0$
\end{enumerate} 
 The \emph{Courant bracket} is defined as the skew-symmetrized Dorfman bracket,
$$ \lb f,g\rb_C = \frac12\big(\lb f,g\rb-\lb g,f\rb\big) = \lb f,g\rb-\frac 12 \partial \langle f,g\rangle$$
and  satisfies 
$$ \lb f,\lb g,h\rb_C\rb_C+\lb g,\lb h,f\rb_C\rb_C +\lb h, \lb f,g\rb_C \rb_C = d\ \text{Nij}(f,g,h) $$
where the Nijenhuis operator on the right hand side is given by
$$\text{Nij}(f,g,h)= \frac 13\big(  \langle \lb f,g\rb_C ,h\rangle + \langle \lb g,h\rb_C ,f\rangle + \langle \lb h,f\rb_C ,g\rangle\big)\ .$$
The Courant-Dorfman algebra is called \emph{non-degenerate} if the map
$$ V \to Hom_R(V,R), \ \ f \mapsto \langle f, \cdot\rangle$$
is an isomorphism. For such non-degenerate pairings, one can show that axioms $(3),(4)$ and $(6)$ are in fact redundant \cite{R}. Our first result is to elucidate the relationship to Poisson vertex algebras.
\begin{prop} \label{thm1}
A Poisson vertex algebra $(\cV, \{\cdot_\lambda\cdot\}, \partial)$ determines a weak Courant-Dorfman algebra $(\cV, \langle \cdot,\cdot \rangle, \lb\cdot,\cdot\rb, \partial)$ by the following assignment,
\begin{eqnarray*} 
\lb f,g\rb &=& \{f_\lambda g\}|_{\lambda =0} \\
\langle f,g\rangle &=& \sum_{j=1}^\infty \frac{(-\partial)^{j-1}}{j!}\frac{d^j}{d\lambda^j}\Big(\{f_\lambda g\}+\{g_\lambda f\}\Big)\Big|_{\lambda =0}
\end{eqnarray*}
\end{prop}

\noindent \begin{proof} We need to check that  axioms (1), (2) and (3) are satisfied, where in this case the vector spaces $V$ and $R$ are both equal to $\cV$. First it is  convenient to expand the $\lambda$-bracket in so called $j$-th products \cite{K},
$$ \{f_\lambda g\} = \sum_{j\in \ZZ_+} (f_{(j)}g) \frac{\lambda^j}{j!}$$
where $ f_{(j)}g := \frac{d^j}{d\lambda^j}\{f_\lambda g\}\big|_{\lambda=0}$ for $j\in\ZZ_+$. In other words the $\lambda$-bracket
is a generating function for the non-negative $j$-th products. For negative values these products are extended by using the derivation $\partial$, 
$$ f_{(-j-1)}g := (\partial^j f)g, \ \ \ j\in\ZZ_+\ ,$$
so in particular $f_{(-1)}g = f g$ is the associative commutative product on $\cV$. The properties of the $\lambda$-bracket are then nicely encoded in the \emph{Borcherds identity}
\begin{equation}\label{borcherds}
  \sum_{j\in\ZZ_+} (-1)^j\binom{n}{j}\Big(f_{(m+n-j)}(g_{(p+j)}h)-(-1)^n g_{(n+p-j)}(f_{(m+j)}h)\Big) =
\end{equation} 
$$= \sum_{j\in\ZZ_+}\binom{m}{j}\Big(f_{(n+j)}g)_{(m+p-j)}h) \ $$ 
for all $m,n,p\in \ZZ$. It should be emphasised that the locality condition on the functions ensures that all the sums are finite. Furthermore, we note that 
$$\lb f,g\rb = f_{(0)}g$$
while 
$$\lb g,f\rb = -\{f_{-\lambda-\partial} g\}|_{\lambda =0}  =- \sum_{j\in \ZZ_+} \frac{(-\partial)^j}{j!} (f_{(j)}g) \ .$$ 
Together these imply axiom $(2)$,
$$\lb f,g\rb +\lb g,f\rb  = \sum_{j=1}^\infty \frac{(-1)^{j+1}}{j!} \partial^j(f_{(j)}g) =  \partial  \langle f,g\rangle \ .$$
The Jacobi identity (1) for the Dorfman bracket follows immediately from Borcherds identity for the values $m=n=p=0$. In particular, the $0$-th product is a derivation of all $j$-th products and hence of the bilinear form. Finally, the third axiom
$$\lb \partial f,g\rb=0$$
follows by the translation covariance of Poisson vertex algebras,
$$ (\partial f)_{(j)}g= -j f_{(j-1)}g$$ 
for $j=0$, which in turn can be derived from the Borcherds identity, or equivalently by the sesquilinearity of the $\lambda$-bracket.
\end{proof}
\begin{remark} The Proposition actually holds for any Lie conformal algebra, since the multiplicative structure on $\cV$ is not being used, and consequently for any vertex algebra. We shall return to the definition of vertex algebras in Section \ref{vertex}. 

It has been brought to our attention that a similar result has been obtained independently by J. Ekstrand in his PhD thesis \cite{E}, which includes a further refinement. Namely, since $\cV$ is a commutative differential algebra and carries a natural left $\cV$-module structure, it would become a Courant-Dorfman algebra if the remaining axioms (4), (5) and (6) were satisfied. In \cite{E} Theorem 4.1, a one-to-one correspondence is established between Courant-Dorfman algebras and $\mathbb Z_2$-graded Poisson vertex algebras, $\cV= \cV_0\oplus \cV_1$, where the grading encodes the \emph{conformal weight} of the elements. Under the identification $\cV_0 = R$ and $\cV_1 = V$, axioms (4) and (5) follow from the Leibniz identity and the Jacobi identity respectively, and (6) always holds by the sesquilinearity of the $\lambda$-bracket. 
\end{remark}

We conclude that the $\lambda$-bracket \eqref{l-bracket} determines a weak Courant-Dorfman algebra structure on the space of local functions $\cV_{loc}$.  On the other hand it was noted in \cite{EZ} that by rewriting the Lie bracket between local distributions in the following form,
\begin{equation} \label{current}
\{J_\phi(f),J_\psi(g)\} = J_{\phi \psi}(\lb f,g\rb) + c(f,g; \phi, \psi) \ ,
\end{equation}
this also induces a weak Courant-Dorfman algebra structure on $\cV_{loc}$. The anomalous Schwinger term is given by  
\begin{equation} \label{schwinger}
c(f,g; \phi, \psi) = \sum_{j=1}^\infty \int_{S^1} \psi(s) (\partial_t^j\phi)(s) C_j(f,g) ds
\end{equation}
for some functions $C_j\colon\cV_{loc}\times \cV_{loc} \to \cV_{loc}$ and is subject to the normalization $c(f,g; 1, \psi)=0$. Below we show that these are in fact the same structures and we obtain thereby an explicit expression for the Dorfman bracket and the functions $C_j$ in \cite{EZ} in the language of Poisson vertex algebras. 

\begin{prop} \label{prop1}
The Dorfman bracket on $\cV_{loc}$ defined by \eqref{current} coincides with the $\lambda$-bracket \eqref{l-bracket} evaluated at $\lambda=0$,
$$\lb f,g\rb = \{f_\lambda g\}\big|_{\lambda=0} = \sum_{i,j\in\cI, n\in\ZZ_+} \frac{\partial g}{\partial u_j^{(n)}}\partial^n \{u_{i\partial}u_j\}_\rightarrow
\frac{\delta f}{\delta u_i}
$$ 
and the functions $C_j$ in \eqref{schwinger}  correspond to the higher derivatives of the $\lambda$-bracket,
$$ C_j(f,g) = \frac{1}{j!} \frac{d^j}{d\lambda^j}\{f_\lambda g\}\big|_{\lambda=0}= \frac{f_{(j)}g}{j!}, \ \ j\geq 1 \ .$$ 

\end{prop}
\begin{proof} The assertion follows by a simple Taylor expansion argument. First let us write out the left hand side of \eqref{current},
$$\{J_\phi(f),J_\psi(g)\}  = \int_{S^1\times S^1} \phi(t)\psi(s) \{f(t),g(s)\}dt ds \ .$$
Inserting into this the Taylor expansion of the test function,
$$\phi(t) = \sum_{j\in\ZZ_+}(\partial_t^j\phi)(s)\frac{(t-s)^j}{j!}$$  
we read off the following expression for the Dorfman bracket
$$ \lb f,g\rb(s) = \int_{S^1}\{f(t),g(s)\} dt $$ 
and for the functions
$$ C_j(f,g)(s) =  \int_{S^1} \frac{(t-s)^j}{j!}\{f(t),g(s)\} dt \ .$$
On the other hand,  by expanding the exponential in the Fourier transform \eqref{fourier},
$$ \{f(s)_\lambda g(s)\} = \int_{S^1} e^{\lambda (t-s)}\{f(t),g(s)\}dt = \sum_{j\in\ZZ_+} \frac{\lambda^j}{j!}\int_{S^1} (t-s)^j\{f(t),g(s)\} dt $$
the result follows by a simple comparison.  
\end{proof}
\begin{remark} Notice that since the anomaly terms $C_j$ are determined by the $j$-th products, they are all intertwined with each other and with the Dorfman bracket via relations dictated by the Borcherds identity \eqref{borcherds}.
\end{remark}
\section{T-duality of Alekseev-Strobl current algebras}
\label{sec:duality}
\subsection{Alekseev-Strobl currents}
In \cite{AS} the authors introduced a special class of local distributions which are parametrised by smooth sections of the generalised tangent bundle $TE\oplus T^*E$. Recall that the generalised tangent bundle is  equipped with a natural Courant-Dorfman algebra structure $(V, R, \langle \cdot,\cdot \rangle, \lb\cdot,\cdot\rb_H, \partial)$, where $V = \Gamma(TE\oplus T^*E)$, $R=C^\infty(E)$,  $\partial = d$ is the usual de Rham differential,  $\langle \cdot,\cdot \rangle$ is defined by the canonical non-degenerate pairing of $TE$ and $T^*E$,
$$ \langle (\xi,\alpha),(\chi,\beta) \rangle = \frac12(\iota_\chi \alpha + \iota_\xi\beta)$$
and the Dorfman bracket is given by
$$\lb (\xi,\alpha),(\chi,\beta)\rb_H = ([\xi,\chi], \CL_\xi\beta - \iota_\chi d\alpha + \iota_\xi \iota_\chi H) $$
where $[\cdot,\cdot]$ is the Lie bracket on vector fields and $H\in \Omega^3_\ZZ(E)$. In other words this is an  \emph{exact Courant algebroid} with the anchor map $TE\oplus T^*E\to TE$  given by the natural projection \cite{LWX}.
More importantly, this Courant-Dorfman algebra structure carries over \emph{pointwise} to sections of the looped generalized tangent bundle $L(TE\oplus T^*E) \to LE$. Using the natural inclusion \eqref{inclusion} we have, for every $p\in T^*_xLE  =\Gamma(S^1, x^*T^*E)$ in the fibre over $x\in LE$, a map 
$$LE \to LTE\oplus LT^*E, \ x \mapsto (\partial x, p) \ .$$ 
\begin{defn}
An \emph{Alekseev-Strobl function} is  a local function in $C^\infty(J^1(LT^*E))$ of the form
\begin{equation}\label{AS} 
f_{(\xi,\alpha)}(\partial x, p) =  \langle (\xi, \alpha),(\partial x, p)\rangle = \iota_{\partial x} \alpha + \iota_\xi p
\end{equation}
where $(\xi, \alpha) \in \Gamma(TE\oplus T^*E)$ is extended pointwise to a section of $LTE\oplus LT^*E$. Let $\cV_{AS}\subset \cV_{loc}$ denote the Poisson subalgebra generated by the Alekseev-Strobl  functions.
\end{defn}
The associated Lie subalgebra of local distributions will be denoted by $\cD_{AS} \subset \cD$ and its elements referred to as \emph{Alekseev-Strobl currents}, since special cases include the current algebras of the Wess-Zumino-Witten model and the Poisson $\sigma$-model. Similarly we write $\cF_{AS}\subset \cF^0$ for the Lie subalgebra of \emph{Alekseev-Strobl functionals}.

An interesting observation made in \cite{AS} is that there is a natural correspondence between the Courant-Dorfman algebra structure on $\cV_{AS}$ determined by the twisted symplectic form $\omega_H \in \Omega^2_\ZZ(T^*LE)$ and the $H$-twisted Courant algebroid $\Gamma(TE\oplus T^*E)$ described above. Below we rederive this result in the language of Poisson vertex algebras.
\begin{proposition}{(\cite{AS})} \label{prop2}
The Courant-Dorfman algebra structure on $\cV_{AS}$, induced by the Poisson bracket associated to the twisted symplectic form \eqref{symp}, satisfies
\begin{eqnarray*} 
\lb f_{(\xi,\alpha)} , f_{(\chi,\beta)}\rb &=& - f_{\lb (\xi,\alpha), (\chi,\beta) \rb_H} \\
\langle f_{(\xi,\alpha)} , f_{(\chi,\beta)} \rangle &=& 2 \langle (\xi,\alpha),(\chi,\beta)  \rangle
\end{eqnarray*}
\end{proposition}
\begin{proof} Choosing local Darboux coordinates $u_i=(x^i,p_i)$ on $T^*E$, we use formula \eqref{poisson} to compute the $j$-th products between the Alekseev-Strobl functions. First, we have
\begin{eqnarray*}
f_{(\xi,\alpha)(0)} f_{(\chi,\beta)}(s)&=& \int_{S^1}\{f_{(\xi,\alpha)}(t) , f_{(\chi,\beta)}(s)\}dt\\
&=& \sum_{\substack{i,j\in \cI \\ m,n \in \ZZ_+}} \int_{S^1}\frac{\partial f_{(\xi,\alpha)}(t)}{\partial u_i^{(m)}} \frac{\partial f_{(\chi,\beta)}(s)}{\partial u_j^{(n)}} \partial^m_t\partial^n_s \{u_i(t),u_j(s)\} dt \ .\\ 
\end{eqnarray*}
By \eqref{AS} it follows that the only non-vanishing derivatives are
$$ \frac{\partial f_{(\xi,\alpha)} }{\partial x^i}= \sum_{k=1}^N \frac{\partial\alpha_k(x)}{\partial x^i}\partial x^k + \frac{\partial\xi_k(x)}{\partial x^i}p_k, \ \ \ \ \frac{\partial f_{(\xi,\alpha)} }{\partial p_i} = \xi^i(x), \ \ \ \ \frac{\partial f_{(\xi,\alpha)} }{\partial (\partial x^i)}=\alpha_i(x) \ ,$$
and similarly for $f_{(\chi,\beta)}\in\cV_{AS}$. Inserting these into the right hand side above and using the brackets \eqref{pb}, a tedious but straightforward calculation yields
\begin{eqnarray*} 
f_{(\xi,\alpha)(0)} f_{(\chi,\beta)}&=& \sum_{i,k=1}^N \left(\chi^i\frac{\partial \xi^k}{\partial x^i}-\xi^i\frac{\partial \chi^k}{\partial x^i}\right)p_k\\
&&+\sum_{i,k=1}^N\left(\chi^i\frac{\partial \alpha_k}{\partial x^i}-\xi^i\frac{\partial \beta_k}{\partial x^i} - \frac{\partial (\alpha_i\chi^i)}{\partial x^k}-\sum_{j=1}^N\xi^i\chi^jH_{ijk}\right) \partial x^k\\
&=& - f_{\lb (\xi,\alpha), (\chi,\beta) \rb_H}
\end{eqnarray*}
Next we have
$$ f_{(\xi,\alpha)(1)} f_{(\chi,\beta)}(s)= \int_{S^1}(t-s)\{f_{(\xi,\alpha)}(t) , f_{(\chi,\beta)}(s)\}dt \ .$$
Repeating the same procedure, this time only a few terms survive due to the presence of $(t-s)\delta(t-s)$ in the integrand and we are left with
$$ f_{(\xi,\alpha)(1)} f_{(\chi,\beta)}= \sum_{k=1}^N \alpha_k\chi^k + \beta_k\xi^k = 2 \langle (\xi,\alpha),(\chi,\beta)  \rangle$$
It follows that the $1$-st product is actually symmetric in this case. Similarly it is not hard to see that the term $(t-s)^j\delta(t-s)$ for $j\geq 2$ leads to the vanishing of all higher $j$-th products.

Now the knowledge of all  non-negative $j$-th products on $\cV_{AS}$ amounts to knowing the $\lambda$-bracket and the result  follows by applying Proposition \ref{thm1}. 
\end{proof}
\begin{remark} Notice that we have omitted the prefix `weak'. This is because the axioms of the Courant algebroid $\Gamma(LTE \oplus LT^*E)$ translate into those of a Courant-Dorfman algebra on $\cV_{AS}$ under the above correspondence. 
  \end{remark}
  
\subsection{T-duality transform} 
In this section we show that the T-duality relations introduced in \cite{BEM1, BEM2} and \cite{CG} establish an isomorphism of the Alekseev-Strobl algebras. 

Let $E\xrightarrow{\TT}M$ be a principal circle bundle equipped with a connection 1-form $A\in \Omega^1(E,\RR)$ and background flux $H\in \Omega_\ZZ^3(E)$, which we can assume without loss of generality is $\TT$-invariant. The curvature  
$F=dA$ provides a real representative of the first Chern class $c_1(E)$ in the de Rham cohomology of $M$. 
Define $\widehat E\xrightarrow{\widehat\TT}M$ to  be the principal circle bundle with $c_1(\widehat E)$ represented by the 2-form $\int_\TT H$ and choose a connection $\hat A$ such that the curvature $\hat F = d\hat A$
has the property that $\hat F = \int_\TT H. $ This is always possible by geometric pre-quantization \cite{K}. Consider the correspondence 
space commutative diagram,
$$
\qquad\xymatrix @=4pc @ur 
{ (E,H) \ar[d]_\pi & E\times_M \widehat E \ar[d]^{\hat p} \ar[l]_p \\ M & (\widehat{E},\hat{H}) \ar[l]^{\hat \pi}}
$$
Then
$$
  H  = A\wedge \hat F - \pi^* \Omega \,,
$$
for some $\Omega\in \Omega^3(M)$,   while the T-dual $\hat H$ is the $\widehat\TT$-invariant integral 3-form  given
by
$$
  \hat H  = F\wedge \hat A -  \hat \pi^*\Omega \,.
$$
Disregarding torsion, the T-dual flux is uniquely determined by  the relation
\begin{equation}\label{corr} 
p^*H - \hat{p}^*\hat H = d \cF\ ,
\end{equation}
where $\cF = p^*A\wedge  \hat{p}^*\hat A \in \Omega^2(E\times_M \widehat E)$.
The following transform establishes an isomorphism of twisted de Rham complexes
\begin{align*}
T\colon & (\Omega^\bullet(E)^\TT, d_H)
\cong (\Omega^{\bullet-1}(\widehat E)^{\widehat \TT}, d_{\hat H})\\ 
 \alpha &\mapsto T(\alpha) = \int_\TT e^{\cF}\wedge p^* \alpha
\end{align*}
where $d_H = d-H\wedge$ is the twisted differential, and in particular one has 
\begin{equation}\label{intertwine} 
T\circ d_H = - d_{\hat H}\circ T \ .
\end{equation}
The map $T$ is the smooth analogue of the Fourier-Mukai transform \cite{Mukai} in the case when the flux $H=0$. It was generalized 
by Hori \cite{Hori} to the case when the flux is exact $H=dB$ and defined in general in \cite{BEM1, BEM2}.

The connection $A$ on the circle bundle $E$ determines a looped connection on $LE$ which we shall also denote by $A$. This furnishes us with a splitting
$$ LTE\oplus LT^*E = LTM\oplus L\ft\oplus LT^*M\oplus L\ft^* $$
and similarly on the T-dual manifold. Any element of $\Gamma(LTE\oplus LT^*E)$ is thus of the form $(\xi, \xi_w, \alpha, \alpha_p) $, 
where $\xi$ is horizontal, $\alpha$ is basic and $(\xi_w, \alpha_p)\in L\ft\oplus L\ft^*$. Here we have identified $LT\TT \cong L\ft$ and $LT^*\TT \cong L\ft^*$,  corresponding to the `winding' and `momentum' components respectively. T-duality is an exchange of these quantities, so we define a map
$$\Psi \colon \Gamma(LTE\oplus LT^*E) \to\Gamma(LT\widehat E\oplus LT^*\widehat E)$$
to be the interchange of these entries,
$$ (\xi, \xi_w, \alpha, \alpha_p) \mapsto (\xi, \alpha_p, \alpha, \xi_w)  \ .$$
The induced map  on the space of Alekseev-Strobl functions $\Psi: \cV_{AS} \to \widehat \cV_{AS}$ is thus
$$ f_{(\xi, \xi_w, \alpha, \alpha_p)} \mapsto  f_{(\xi, \alpha_p, \alpha, \xi_w)}  \ .$$
Let $\cV_{AS}^\TT \subset \cV_{AS}$ denote the subalgebra of Alekseev-Strobl functions parametrized by the $\TT$-invariant sections of $TE\oplus T^*E$. We write $\cF_{AS}^\TT$ and $\cD_{AS}^\TT$ for the associated Lie subalgebras of invariant Alekseev-Strobl  functionals and currents respectively.
\begin{thm}\label{T-duality}
 The map $\Psi$ determines an isomorphism of the Courant-Dorfman algebras of invariant Alekseev-Strobl functions,
 $$ (\cV_{AS}^\TT, \langle \cdot,\cdot \rangle, \lb\cdot,\cdot\rb, \partial) \cong  (\widehat \cV_{AS}^{\widehat\TT}, \langle \cdot,\cdot \rangle, \lb\cdot,\cdot\rb, \partial)$$
 \end{thm}

\begin{proof} The idea is to adapt the isomorphism between Courant algebroids \cite{CG} to the looped generalized tangent bundles. 
Firstly, there is a natural action 
$$\Gamma(LTE\oplus LT^*E) \times \Omega^\bullet(LE) \to \Omega^\bullet(LE)$$ 
given by the parity reversing map
$$ (\xi,\alpha)\cdot \omega = \iota_\xi \omega + \alpha\wedge \omega \ .$$
This further extends to an action by the Clifford bundle $Cl(LTE\oplus LT^*E)$ due to
\begin{equation}\label{clifford} 
(\xi,\alpha)\cdot ((\xi,\alpha)\cdot \omega) = \langle (\xi,\alpha),(\xi,\alpha)\rangle \omega \ ,
\end{equation}
so the space of differential forms on $LE$ becomes an irreducible spin module. 

Secondly, any $\TT$-invariant form on $LE$ obtained by looping an invariant form in $\Omega^\bullet(E)^\TT$ can be written as 
$$\omega = \pi^*(\alpha) + A\wedge \pi^*(\beta) $$
where $\pi\colon LE \to LM$, $\alpha,\beta \in \Omega^\bullet(LM)$ and $A$ is the looped connection on $LE$. We note that this is \emph{not} true for any invariant form in $\Omega^\bullet(LE)^\TT$. The map $T\colon\Omega^\bullet(E)^{\TT}\to \Omega^{\bullet+1}(\widehat E)^{\widehat \TT}$ extends pointwise to the looped bundles and sends $\omega$ to
$$T(\omega) = \hat \pi^*(\beta)-\hat A\wedge \hat \pi^*(\alpha) $$
where $\hat \pi\colon L\widehat E \to LM$. If $(\xi,\alpha)$ and $\omega$ are invariant, a straightforward calculation shows that
\begin{equation}\label{commute}  T((\xi ,\alpha )\cdot \omega) = \Psi((\xi , \alpha )) \cdot T(\omega) \ .\end{equation}

Thirdly, the Dorfman bracket on $\Gamma(TE^*\oplus TE)$ is a derived bracket in the sense that
$$ \lb (\xi,\alpha),(\chi,\beta)\rb_H \cdot \omega  = [[d_H,  (\xi,\alpha)],(\chi,\beta)]\cdot \omega$$ 
and again this extends  pointwise to $\Gamma(LTE^*\oplus LTE)$, as does the intertwining relation \eqref{intertwine}
It is now an easy task to show that the Dorfman bracket on $\Gamma(LTE^*\oplus LTE)$ is preserved by the T-duality map $\Psi$. Namely,
\begin{eqnarray*}
  \Psi(\lb (\xi,\alpha),(\chi,\beta)\rb_H ) \cdot T(\omega)&=&  T(\lb (\xi,\alpha),(\chi,\beta)\rb_H\cdot \omega) \\
  &=& T( [[d_H,  (\xi,\alpha)],(\chi,\beta)]\cdot \omega)\\
  &=& T( [[d_{\hat H},  \Psi((\xi,\alpha))],\Psi((\chi,\beta))]\cdot \omega)\\
  &=&  \lb \Psi((\xi,\alpha)),\Psi((\chi,\beta))\rb_{\hat H} \cdot T(\omega)
\end{eqnarray*} 
which implies  
$$\Psi(\lb (\xi,\alpha),(\chi,\beta)\rb_H )=\lb \Psi((\xi,\alpha)),\Psi((\chi,\beta))\rb_{\hat H} \ .$$
Similarly, the exchange of the bilinear forms
$$ \langle (\xi,\alpha),(\chi,\beta)\rangle= \langle\Psi((\xi,\alpha)),\Psi((\chi,\beta))\rangle $$
is  a simple consequence of the Clifford action \eqref{clifford} and \eqref{commute}. The assertion finally follows by applying Proposition \ref{prop2}.
\end{proof}
We have immediately the following result.
\begin{thm}\label{T-duality2}  The map $\Psi$ extends to an isomorphism of Poisson algebras of invariant Alekseev-Strobl functions,
$$(\cV_{AS}^\TT, \cdot, \{\cdot,\cdot\}) \cong (\widehat \cV_{AS}^{\widehat\TT}, \cdot, \{\cdot,\cdot\})$$
and consequently to an isomorphism of the associated Poisson vertex algebras  and of the Lie algebras of invariant Alekseev-Strobl functionals and currents.
\end{thm}
\begin{proof}
Recall that  the Poisson bracket between Alekseev-Strobl functions is given by
$$\{f_{(\xi,\alpha)}(t) , f_{(\chi,\beta)}(s)\} = -f_{\lb (\xi,\alpha), (\chi,\beta) \rb_H}(s) \delta(t-s) + 2 \langle (\xi,\alpha),(\chi,\beta)  \rangle(s) \partial_t \delta(t-s) \ ,$$
when viewed as distributions. Since the map $\Psi$  preserves the pointwise multiplication of functions, it clearly extends to an isomorphism of Poisson algebras. It is also clear that the $\lambda$-bracket, which is the Fourier transform of the Poisson bracket, is preserved under the T-duality map. 
Finally, the isomorphism of the invariant Alekseev-Strobl current algebras follows by Proposition \ref{prop1} and $\cF_{AS}^\TT\cong \widehat\cF_{AS}^{\widehat\TT}$ is a consequence of the  inclusion $\cF_{AS}^\TT \subset \cD_{AS}^\TT$.
\end{proof}

\begin{remark} An interesting consequence of T-duality is that it induces an isomorphism of the Lie algebra cohomologies of the invariant Alekseev-Strobl current algebras,
$$ H^\bullet(\cD_{AS}^\TT,\RR) \cong H^\bullet(\widehat \cD_{AS}^{\widehat\TT},\RR)$$
\end{remark}

Furthermore, we have the following correspondence diagram for the  phase spaces of the T-dual manifolds,
\begin{equation*}
\qquad\xymatrix @=4pc @ur 
{ (T^*LE,\omega_H) \ar[d]_\pi & T^*LE\times_{T^*LM} T^*L\widehat E \ar[d]^{\hat p} \ar[l]_p \\ T^*LM & (T^*L\widehat E, \omega_{\hat H}) \ar[l]^{\hat \pi}}
\end{equation*}
By \eqref{corr} it follows that the twisted symplectic structures are related by
$$ p^*\omega_H - \hat{p}^*\omega_{\hat H} = d \int_{S^1}\ev^*\Big( (q\circ p)^*A\wedge  (\hat q\circ \hat p)^*\hat A\Big) \ ,$$
where $q$ and $\hat q$ are the projection maps of the cotangent bundles to the loop spaces. 
In fact, in \cite{BHM} it was observed  that $\omega_H$ and $\omega_{\hat H}$ are both obtained by symplectic reduction of a $\TT\times \widehat \TT-$invariant symplectic form on the correspondence space.

Lastly, recall that a \emph{Dirac structure}  is a subbundle $L\subset TE\oplus T^*E$ that is Lagrangian and its space of sections $\Gamma(L)$ is closed under the Courant bracket. As noted in \cite{AS}, Dirac structures correspond to anomaly free current algebras (\emph{i.e.} the vanishing of the Schwinger term in \eqref{current}). Since Dirac structures are interchanged by the map $\Psi$, we conclude that anomaly cancellation is preserved under T-duality. 

\section{Quantization of Alekseev-Strobl current algebras} \label{vertex}
In this section we make some remarks on the quantization of Alekseev-Strobl current algebras and its behaviour under T-duality. To set the stage let us recall the definition of vertex algebras \cite{DSK}.

\begin{definition}
A \emph{vertex algebra} is  a quintuple 
$(V,\vac, \partial,[\cdot_{\lambda}\cdot],:\ :)$ such that
\begin{itemize}
\item[$(i)$] $(V,\partial,[\cdot_{\lambda}\cdot])$ is a Lie conformal algebra,
\item[$(ii)$] $(V,\vac,\partial,:\ :)$ is a unital differential algebra satisfying the \emph{strong quasi-commutativity} relation
\begin{equation}\label{q-com}
 :a:bc::-:b:ac::=:\left(\int_{-\partial}^0[a_{\lambda}b]d\lambda\right)c:
\end{equation}
\item[$(iii)$] The $\lambda$-bracket $[\cdot_{\lambda}\cdot]$ and the `normally ordered product' $:\ :$ are related by the \emph{non-commutative Wick formula}
\begin{equation} \label{wick}
[a_{\lambda}:bc:]= :[a_{\lambda}b]c: + :b[a_{\lambda}c]: + \int_0^\lambda [[a_{\lambda}b]_{\mu}c]d\mu \ .
\end{equation}
\end{itemize}
\end{definition}

Next let us consider a family of vertex algebras $(V_{\hbar}, \vac_{\hbar}, \partial_{\hbar}, [\cdot_{\lambda}\cdot]_\hbar, :\ :_{\hbar})$ depending on a formal parameter $\hbar$. In other words, $V_{\hbar}$ is a free module over $\mathbb R[[\hbar]]$ such that $[{V_{\hbar}}_{\lambda}V_{\hbar}]\subseteq \hbar V_{\hbar}$. The \emph{quasiclassical limit} of this family is defined by 
$$V = \lim_{\hbar\to 0}V_{\hbar} := V_{\hbar}/\hbar V_{\hbar}\ , $$ where we denote by $1,\partial,$ and $\cdot$ the images of  $\vac_{\hbar}, \partial_h,$ and $ :\ :_\hbar$ respectively in this quotient. Similarly we write $\{ a_{\lambda} b\} $ for the image of $\frac{[\tilde{a}_{\lambda}\tilde{b}]_\hbar}{\hbar}$, where $\tilde{a}, \tilde{b} \in V_{\hbar}$ are pre-images of $a,b \in V$. These are defined up to a multiple of  $\hbar$ which disappears when we pass to the quotient, so $\{a_{\lambda}b\}$ is independent of the choice of pre-images. 

It is not hard to check that since the integral terms (or `quantum corrections') in \eqref{q-com} and \eqref{wick} are of non-zero order in the parameter $\hbar$,
$$ \int_{-\partial}^0[a_{\lambda}b]d\lambda \in \hbar V_\hbar, \ \ \ \ \ \int_0^\lambda [[a_{\lambda}b]_{\mu}c]d\mu \in \hbar^2 V_\hbar \ ,$$
they vanish as `$\hbar$ tends to zero' and the quasiclassical limit $(V,1,\cdot, \partial,\{\cdot_{\lambda}\cdot\})$ becomes a \emph{Poisson vertex algebra}. Indeed, the commutativity of the product $\cdot$ follows from \eqref{q-com} by setting $c=\vac$. Similarly by rewriting \eqref{q-com}  as
$$ ::ab:c:-:a:bc::=:\int_0^\partial a [b_{\lambda}c]d\lambda: + :\int_0^\partial b[a_{\lambda}c]d\lambda:$$
we conclude that the product becomes associative in the limit $\hbar \to 0$. In other words, a Poisson vertex algebra is a vertex algebra without `quantum corrections'.

Quantization is the inverse operation to the quasiclassical limit. Below we explain a general procedure for quantizing a Lie conformal algebra \cite{K}.  There is a Lie algebra $R_{Lie}$ associated to any Lie conformal algebra $R$ defined by the same underlying vector space and with the Lie bracket
\begin{equation}\label{lie}[a,b] = \int_{-\partial}^0 [a_\lambda b]d\lambda\end{equation}  
for all $a,b \in R$. We construct a family of Lie conformal algebras by setting
\begin{equation}\label{family}
R_\hbar = R, \ \ \ [a_\lambda b]_\hbar := \hbar  [a_\lambda b] \ .
\end{equation}  
Now the associated universal enveloping algebra
\begin{equation}\label{universal}
\cU(R_{\hbar Lie}) = \cT(R)\Big/\left(a\otimes b - b\otimes a -\int_{-\partial}^0 [a_\lambda b]_\hbar d\lambda \right)
\end{equation}  
determines a family of vertex algebras $V_\hbar = \cU(R_{\hbar Lie})$. This follows by the fact that $R_{\hbar Lie}$ is canonically isomorphic to the \emph{creation Lie algebra} $ (\text{Lie}\ R_\hbar)_+$ associated to $R_\hbar$, so the vector space \eqref{universal} inherits the natural universal enveloping vertex algebra structure on 
$$\cU(\text{Lie}\ R_\hbar)/\cU(\text{Lie}\ R_\hbar)(\text{Lie}\ R_\hbar)_- = \cU((\text{Lie}\ R_\hbar)_+) \ .$$
Here $\text{Lie}\ R_\hbar=(\text{Lie}\ R_\hbar)_-\oplus (\text{Lie}\ R_\hbar)_+$ is the unique $\partial$-invariant splitting of the Lie algebra $\text{Lie}\ R_\hbar = R_\hbar[t,t^{-1}]/(\partial + \partial_t)R_\hbar[t,t^{-1}]$ with the bracket
$$ [at^m,bt^n] = \sum_{j=0}^m \binom{m}{j}(a_{(j)}b)t^{m+n-j}\ .$$
The vacuum vector $\vac\in V_\hbar$ corresponds to the image of $1\in \cU(\text{Lie}\ R_\hbar)$ and $\partial$ extends to $V_\hbar$ by derivations. The quasiclassical limit of $V_\hbar$ is the symmetric algebra
$$ V = \lim_{\hbar \to 0} V_\hbar = \cT(R)/(a\otimes b - b\otimes a) = \cS(R)$$
with its associative commutative product and with the $\lambda$-bracket $\{a_\lambda b\} = [a_\lambda b]$ for $a,b\in R$,  extended to $\cS(R)$ by  left and right Leibniz rule.

Returning to the invariant Alekseev-Strobl Poisson vertex algebra $(\cV^\TT_{AS},\{\cdot_\lambda\cdot\},\partial)$, we can apply the above described  quantization   to the underlying Lie conformal algebra. Let $V^\TT_{\hbar,AS}$ denote the associated family of vertex algebras. By  combining \eqref{lie}, \eqref{family} and Theorem \ref{T-duality2}, it is clear that the map $\Psi$ constructed in the previous section induces an isomorphism of these families of vertex algebras.

\begin{theorem} T-duality determines an isomorphism of the families of invariant Alekseev-Strobl vertex algebras,
$$(V^\TT_{\hbar,AS}, \vac, \partial, [\cdot_{\lambda}\cdot]_\hbar,:\ :_\hbar) \cong (\widehat{V}^{\hat\TT}_{\hbar,AS}, \vac, \partial, [\cdot_{\lambda}\cdot]_\hbar,:\ :_\hbar) \ . $$ 
\end{theorem}
\begin{remark} Since the Alekseev-Strobl current algebra $\cD_{AS}^\TT$ is fully determined by the Lie conformal algebra structure on $\cV_{AS}^\TT$, the Theorem implies that the quantized current algebras are isomorphic under T-duality. 
Notice that by the argument above we have the following commutative diagram, 
\begin{equation*} 
\xymatrix{
\left(V, \cdot, \{\cdot_\lambda\cdot\}\right) \ar[r]_\cong \ar[d]_{Quantization}
& \left(V, \cdot, \{\cdot_\lambda\cdot\}\right)  \ar[d]^{Inclusion} \\
\ar[r]^{\hbar\to 0}
\left(V_\hbar, :\ :_\hbar, \{\cdot_\lambda\cdot\}_\hbar\right) & \,  \left(S(V), \circ, \{\cdot_\lambda\cdot\}\right)  .
}
\end{equation*}
The horizontal maps can be viewed as an analogue of the symbol maps (of differential operators) and
the Poincar\'e-Birkhoff-Witt Theorem in this context.

\end{remark}


\end{document}